\documentclass[%
aps,prl,twocolumn,superscriptaddress]{revtex4-1}
\usepackage{amsmath,amssymb,graphicx,dcolumn,bm,paralist}

\usepackage{epsfig}
\usepackage{subfig}

\usepackage[colorlinks,urlcolor=blue,citecolor=blue,linkcolor=blue]{hyperref}

\newcommand{\ignore}[1]{}

\newtheorem{theorem}{Theorem}

\newtheorem{claim}[theorem]{Claim}

\newtheorem{definition}[theorem]{Definition}
\newtheorem{prop}[theorem]{Proposition}

\newtheorem{corollary}[theorem]{Corollary}

\newcommand{\cB}{{\cal B}}

\newcommand{\cD}{{\cal D}}

\newcommand{\cF}{{\cal F}}

\newcommand{\cN}{{\cal N}}

\newcommand{\cT}{{\cal T}}

\newcommand{\sgn}{\mbox{sign}}

\newcommand{\RR}{\mathbb{R}}

\newcommand{\wh}[1]{\widehat{#1}}
\newcommand{\bal}{\delta}
\newcommand{\biasinf}[3]{{\Inf_{#2}}({#3};{#1})}
\newcommand{\biastot}[2]{\I(#2;{#1})}
\newcommand{\dm}{K}

\newcommand{\pr}{{\rm Pr}}

\newcommand{\qed}{\hfill $\Box$}

\newcommand{\EX}{\hbox{\bf E}}

\newcommand{\flip}[2]{#1^{(#2)}}
\newcommand{\Inf}{{\rm{Inf}}}

\newcommand{\I}{\mathrm{I}}
\newcommand{\Iex}{\mathrm{I}}

\newcommand{\nbdt}[2]{N^-_{\leq #1,#2}}

\newcommand{\maj}{MAJ}
\newcommand{\al}{AND}

\newcommand{\Sec}[1]{\hyperref[sec:#1]{\S\ref*{sec:#1}}} 
\newcommand{\App}[1]{\hyperref[sec:#1]{Appendix~\ref*{sec:#1}}} 
\newcommand{\Eqn}[1]{\hyperref[eq:#1]{(\ref*{eqn:#1})}} 
\newcommand{\Fig}[1]{\hyperref[fig:#1]{Figure\,\ref*{fig:#1}}} 
\newcommand{\Tab}[1]{\hyperref[tab:#1]{Table\,\ref*{tab:#1}}} 
\newcommand{\Thm}[1]{\hyperref[thm:#1]{Theorem\,\ref*{thm:#1}}} 
\newcommand{\Lem}[1]{\hyperref[lem:#1]{Lemma\,\ref*{lem:#1}}} 
\newcommand{\Prop}[1]{\hyperref[prop:#1]{Prop.~\ref*{prop:#1}}} 
\newcommand{\Cor}[1]{\hyperref[cor:#1]{Cor.~\ref*{cor:#1}}} 
\newcommand{\Def}[1]{\hyperref[def:#1]{Defn.~\ref*{def:#1}}} 
\newcommand{\Alg}[1]{\hyperref[alg:#1]{Alg.~\ref*{alg:#1}}} 
\newcommand{\Ex}[1]{\hyperref[ex:#1]{Ex.~\ref*{ex:#1}}} 
\newcommand{\Clm}[1]{\hyperref[clm:#1]{Claim~\ref*{clm:#1}}} 

\begin{document}

\title{Characterizing short-term stability for Boolean networks over any distribution of transfer
functions}

\author{C. Seshadhri}
\author{Andrew M. Smith}
\affiliation{%
	Sandia National Laboratories, P.O. Box 969, Livermore, California 94551-0969, USA
}%
\author{Yevgeniy Vorobeychik}
\affiliation{Vanderbilt University, Nashville, TN 37235, USA}
\author{Jackson R. Mayo}
\author{Robert C. Armstrong}
\affiliation{%
	Sandia National Laboratories, P.O. Box 969, Livermore, California 94551-0969, USA
}%

\begin{abstract}
We present a characterization of short-term stability of
random Boolean networks under \emph{arbitrary} distributions of transfer functions.  
Given any distribution of transfer functions for a random Boolean network,
we present a formula that decides whether short-term chaos (damage spreading) will happen.
We provide a formal proof for this formula, and empirically show that its predictions are accurate.
Previous work only works for special cases of balanced families. It has been observed
that these characterizations fail for unbalanced families, yet such families are widespread in real
biological networks. 
\end{abstract}

\maketitle

\section{Introduction}

Living systems composed of a wide variety
of cells, genes, or organs operate with uncanny synchrony and
stability, as do numerous engineered and social
systems. 
In a series of seminal papers, Kauffman introduced \emph{Boolean networks}
to study such systems: this abstraction involves a
network representing connectivity, and a family of
Boolean functions determining states of network nodes to model dynamic
behavior~\cite{Kauffman69,Kauffman74}.
Boolean networks have been used to model numerous dynamical
systems, including genetic regulatory networks~\cite{Kauffman69} and
political systems~\cite{Aldana03b}, and have 
received much theoretical attention~\cite{Harris02,Shmulevich04,Moreira05,Aldana03,Kauffman03,Shmulevich03,Mozeika11,Seshadhri11,Squires12}.

A Boolean network has a set of $n$ nodes linked to each other by a directed graph $G$.
Each node $i$ has a Boolean state in $\{-1,+1\}$, an in-degree $K_i$,
and an associated Boolean function
$f_i:\{-1,+1\}^{K_i} \rightarrow \{-1,+1\}$, termed \emph{transfer
  function}.
If the state of node $i$ at time $t$ is $x_i(t)$, its state at time
$t+1$ is described by
\(
x_i(t+1) = f_i(x_{i_1}(t),\ldots,x_{i_{K_i}}(t)).
\)
For the sake of analysis, it is common to study a randomized ensemble
of Boolean networks. The graph $G$ is a directed Erd\H{o}s-R\'{e}nyi network,
where each each vertex $i$ chooses $K_i$ in-neighbors uniformly
at random. There is an underlying distribution (or family)
of Boolean transfer functions $\cF$. Each vertex $i$ independently chooses
the transfer function $f_i$ from $\cF$.

A key parameter of interest is the \emph{short-term stability} of the Boolean network.
Specifically, if a single node has its
state flipped, does the effect of this perturbation die out (quiescence),
exponentially cascade over time (chaos), or is the system right in between (criticality)?
There have been numerous
empirical and mathematical observations about the characteristics of
critical transition points in classes of Boolean
networks~\cite{Harris02,Shmulevich04,Moreira05,Aldana03,Kauffman03,Shmulevich03,Kesseli05,MoSaRa09,MoSaRa10,Seshadhri11,Squires12},
These results require $\cF$ to have specific properties: for example, each truth table
entry is i.i.d. or that functions are \emph{balanced} (number of $+1$ and $-1$ outcomes
is the same) on average.

These are severe restrictions. Various classes of functions 
occur naturally in biological
and social applications, but do not satisfy either of these conditions. 
For example, Kauffman proposed a family of \emph{canalyzing functions}~\cite{Kauffman74}.
A canalyzing function has at least one
input, and one value of that input, that fully determines the output
of the function.
Kauffman observed that many elements of genetic regulatory systems
have like nested canalyzing functions~\cite{Kauffman74,Harris02,Kauffman03,Kauffman04}. 
Previous formal analyses do not yield precise characterizations of short-term stability for such families.

\emph{Threshold functions} also occur in understanding processes on social and biological networks~\cite{Sc78,Gr78,KeKlTa03,LiLo04,DaBo08,ZaAl11,TrMc13}
A threshold function is of the form $f(x_1, x_2,\ldots,x_K) = \sgn(\sum_i c_i x_i - \Theta)$,
where $c_i$s and $\Theta$ are constants. Often there is a bias towards a particular state, so
these previous characterizations fail to predict the critical threshold~\cite{Seshadhri11}.

Our main result gives an exact formula for predicting the short-term dynamics of Boolean networks, for \emph{any} distribution of transfer functions.
We stress that our results are for `semi-annealed' setting. Once we choose the topology of the Boolean network
and the transfer functions from the appropriate distribution, we assume it is fixed. (We do not change
these for each time-step, as in an annealed approximation.) All we need
from the topology is a local tree-structure (as proved in~\cite{Seshadhri11} and subsequently used
in~\cite{Squires12}), which is guaranteed with high probability for Erd\H{o}s-R\'{e}nyi random graph distributions.

While no previous result provides such a formula, our work is closely related to the 
following. Mozeika and
Saad~\cite{MoSaRa09,MoSaRa10,Mozeika11} give a powerful generating function
framework for analysis of Boolean networks, but do not 
characterize short-term stability.
Seshadhri et al.~\cite{Seshadhri11}
introduced the notion of \emph{influence} $\I(\cF)$ of
transfer function distribution $\cF$, an easily computable quantity that determines
the short-term behavior for a highly restricted class of \emph{balanced}
  families $\cF$: on average, functions in $\cF$ are equally likely to output
$+1$ and $-1$.

\vspace{-10pt}

\section{Preliminaries}

We are interested in the sensitivity of a Boolean network state $x(t)
= \{x_1(t)\ldots,x_n(t)\}$ to a small
initial perturbation. Formally, consider the following experiment.
Suppose that a Boolean network starts from state $x$, and after $t$
steps reaches a state $F_t(x)$. 
Now, consider another initial state, $x^{(i)}$ which only differs from $x$
in the $i$th bit. 
Let $H_t$ be the expected Hamming distance between $F_t(x)$ and $F_t(x^{(i)})$,
where $x$ is drawn from some specified (typically uniform) distribution.
How does $H_t$ evolve with time? If $H_t$ can be expressed as $e^{\lambda t}$, then $\lambda$
is the Lyapunov exponent. If $\lambda < 0$, the boolean network
is quiescent; if $\lambda > 0$, the network is chaotic.

We provide some notation and definitions.

\begin{asparaitem}
	\item \textbf{Biased distributions:} We use $\cD_\rho$ to denote the distribution over $\{-1,+1\}$ where the probability of $1$ is $(1+\rho)/2$.
	We choose this notation because the expected value is exactly $\rho$, the bias. 
	Abusing notation, for $y \in \{-1,+1\}^K$, we say $y \sim \cD_\rho$ when each coordinate of $y$ is chosen i.i.d. from $\cD_\rho$.
	\item \textbf{Imbalance:} The \emph{imbalance} of the Boolean network at time $t$, denoted by $\delta_t$, is $\sum_{i=1}^n x_i(t)/n$.
	Informally, this measures the difference between the $+1$s and $-1$s in the network. Observe that if the starting
	state $x(0)$ is chosen from $\cD_\rho$, then $\delta_0 = \rho$. 
\end{asparaitem}
\smallskip
We use tools from 
\emph{harmonic analysis of Boolean functions}, pioneered by Kahn, Kalai, and Linial~\cite{Kahn88}.
The convention in this field is that $-1$ denotes TRUE and $+1$ is FALSE (so multiplication in $\{-1,+1\}$
maps to XOR of $\{0,1\}$ bits). 
Consider $f:\{-1,+1\}^K \rightarrow \{-1,+1\}$, where we think of $f$ as one of the transfer functions.
The standard representation is as a truth table,
with $2^K$ entries in $\{-1,+1\}$. An alternative representation is as a linear combination
of \emph{basis functions}. In the following, we use $y \in \{-1,+1\}^K$ to denote an input
to the transfer function. We use $[K]$ for set $\{1,2,\ldots,K\}$, which denotes
the input coordinates. Refer to~\cite{OD14} for details on the following.
\smallskip
\begin{asparaitem}
	\item \textbf{Parity functions:} For any subset $S$ of coordinates in $[K]$, $\prod_{i \in S} y_i$ is the \emph{parity} on $S$.
	(For $S = \emptyset$, we set the parity to be $1$.)
	\item \textbf{Fourier representation:} Any Boolean function $f$ can be expressed as 
	$f(y) = \sum_{S \subseteq [K]} \widehat{f}(S) \prod_{i \in S} y_i$, where $\widehat{f}(S)$ are called Fourier coefficients.
	This expansion represents $f$ as a multilinear polynomial
	over the Boolean variables $y_1, \ldots, y_K$. 
	It can be shown that $\widehat{f}(S) = 2^{-K} \sum_y f(y)\prod_{i \in S} y_i$, the correlation between
	$f$ and the parity on $S$. 
	(The Fourier coefficients are the Walsh-Hadamard transform 
	of the truth table.) There are exactly $2^K$ different Fourier coefficients, one for each subset of the $K$ inputs.
	For example, consider $K = 2$, and the $AND$ function. A calculation
	yields $AND(y_1, y_2) = 1/2 + y_1/2 + y_2/2 - y_1y_2/2$.
	\item \textbf{Level sets of coefficients, $\sigma_r$:} Of special interest is $\sigma_r(f) = \sum_{C:|C| = r} \widehat{f}(C)$, where $0 \leq r \leq K$.
	This is simply the sum of coefficients corresponding to sets of size $r$. Note that $\sigma_0(f) = \widehat{f}(\emptyset)$ $= \sum_{y} f(y)$.
	This is exactly the imbalance in the truth table of $f$.
	\item \textbf{Influence:} For any function $f$, the influence of the $i$th variable is denoted 
	$\Inf_i(f) = \pr_{y}[f(y) \neq f(\flip{y}{i})]$ (where the probability is over the uniform distribution
	and $\flip{y}{i}$ is obtained by flipping $y$ at the $i$th bit), and the 
	total influence is $\I(f) = \sum_i \Inf_i(f)$. We will define a biased version of this quantity,
	$\Inf_i(f; \rho) = \pr_{y \in \cD_\rho} [f(y) \neq f(\flip{y}{i})]$, and analogously $\I(f; \rho) = \sum_i \Inf_i(f; \rho)$.
\end{asparaitem}

\vspace{-10pt}

\section{Mathematical results}
\vspace{-10pt}

The proofs of our mathematical results are quite involved, and therefore provided in the supplemental material.
We can derive closed form expressions for the evolution of $\delta_t$ (the expected imbalance at time $t$) 
and $H_t$ (the expected Hamming distance at time $t$ after a single bit perturbation).

The evolution of $\delta_t$ ($t > 0$) is determined by the level sets of coefficients of the transfer functions.
We use $\sigma_r(\cF) = \EX_{f \sim \cF}[\sigma_r(f)]$ and $\I(\cF; \delta) = \EX_{f \sim \cF}[\I(f; \delta)]$.

\begin{theorem} \label{thm:rec} Let initial state $x(0)$ be chosen from $\cD_\rho$ (so $\delta_0 = \rho$).
Then $\delta_t$ evolves according to the polynomial recurrence
$\delta_{t+1} = \sum_{r \geq 0} \sigma_r(\cF) \delta^r_{t}$.
\end{theorem}

An equivalent formulation of the recurrence has been
derived by the generating function method in Mozeika and
Saad~\cite{Mozeika11}, though their approach is completely different (they
do not show a connection with Fourier coefficients).
Our approach proves a clean description
of this recurrence, since $\sigma_r(\cF)$ can be easily computed from $\cF$.

Our main theorem shows how the damage caused by a bit perturbation spreads.

\begin{theorem} \label{thm:hamm} Let $\delta_0, \delta_1, \ldots$ be as given by \Thm{rec}. For $t \leq (\log n)/K$,
$H_t = \prod_{0 \leq h < t} \I(\cF; \delta_h)$.
\end{theorem}

In many situations, $\delta_t$ converges to some $\delta^*$. In that case, $H_t \approx [\I(\cF;\delta^*)]^t$.
The Lyapunov exponent is $\log \I(\cF;\delta^*)$, so we get a critical point at $\I(\cF; \delta^*) = 1$.
Our formula gives a provable characterization of short-term stability, for \emph{any} transfer function family $\cF$.

{\bf Balanced families:} As a warmup, we derive previous results that only held for balanced families $\cF$.
In such families, the expected difference (over $\cF$) between $+1$'s and $-1$'s in the transfer functions is exactly zero. This contains the classic random families of Kauffman.
For such a family, $\sigma_0(\cF) = \EX_{f \sim \cF} [\sigma_0(f)] = 0$.
The starting distribution is given by $\cD_0$, so $\delta_0 = 0$. 
Regardless of the values of $\sigma_r(\cF)$ (for $r > 0$),
by \Thm{rec}, $\delta_t = 0$ for all $t$.
Hence, $H_t = [\I(\cF;0)]^t$, and $\I(\cF;0) = 1$ is the critical threshold. This is exactly the main result of~\cite{Seshadhri11}.

\vspace{-10pt}

\section{Applications}

\begin{figure*}[ht!]
\centering
\subfloat[Imbalance for different $\beta$ as $\rho$ increases for threshold function distribution]{
\includegraphics[width=2.25in]{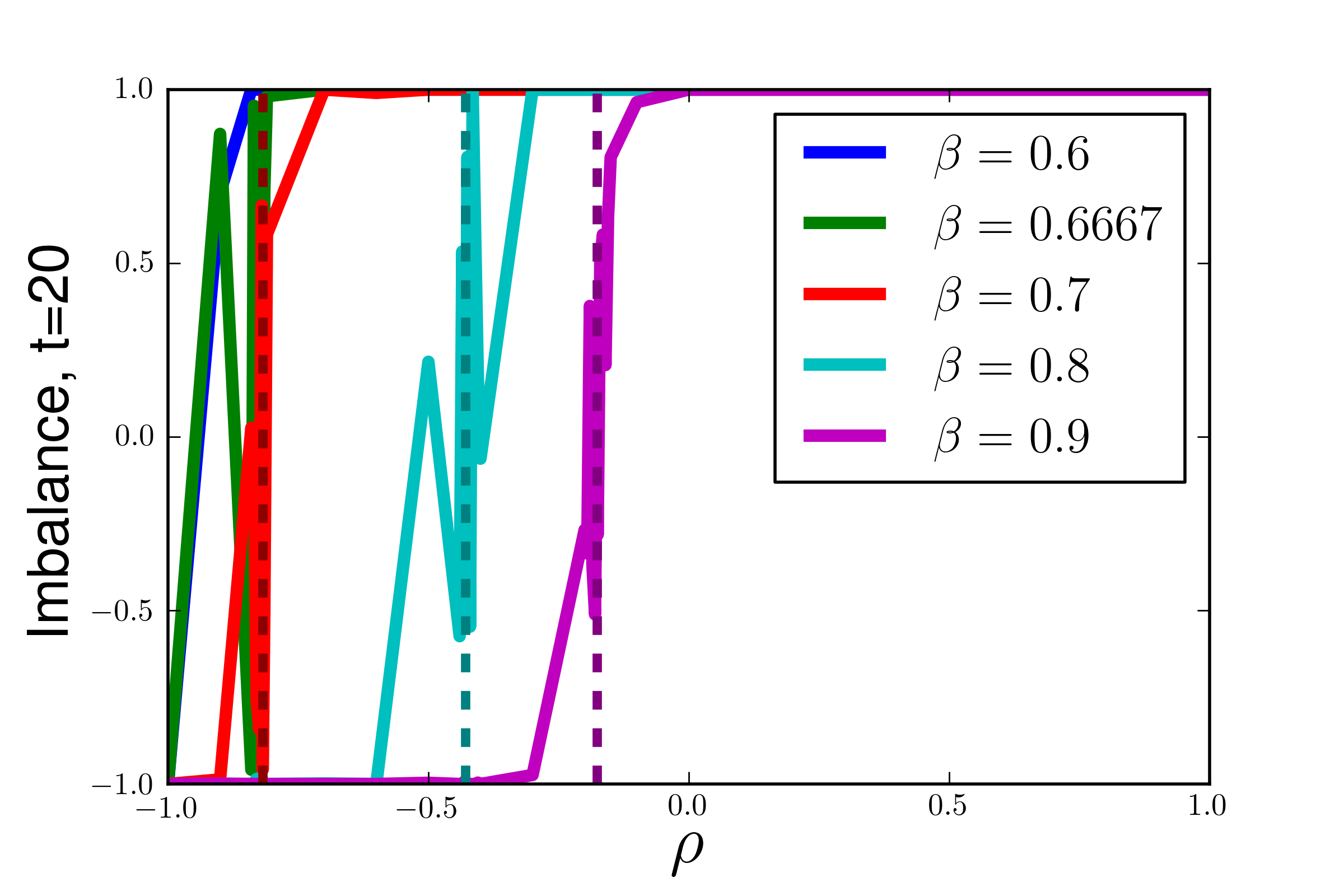} \label{fig:maj-and}
}
\subfloat[Imbalance for different $\rho$ as time increases for nested canalyzing distribution]{
\includegraphics[width=2.25in]{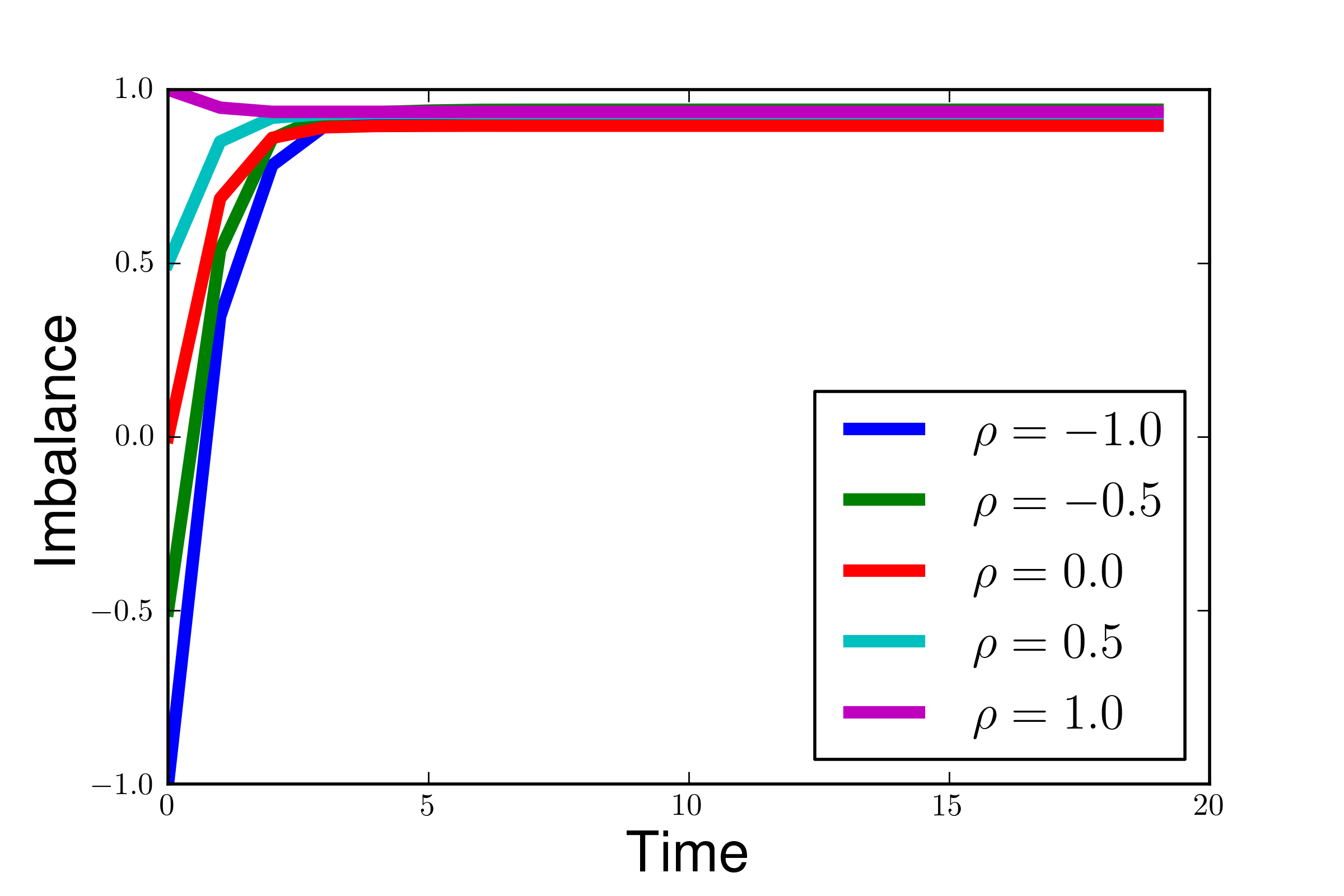} \label{fig:nested-imb}
}
\subfloat[Hamming distance as time increases for nested canalyzing distribution]{
\includegraphics[width=2.25in]{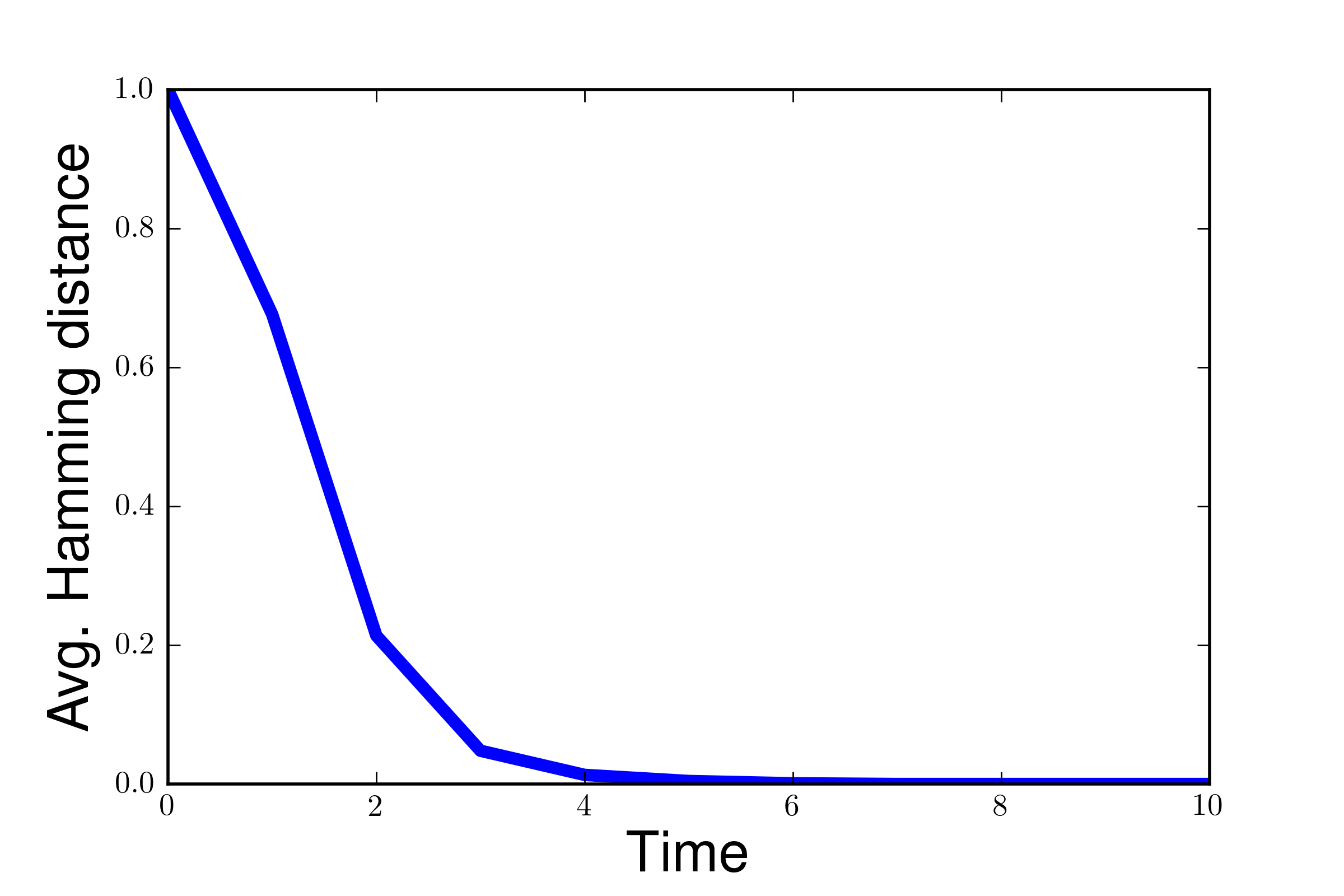} \label{fig:nested-hamm}
}
\caption{Experimental results}
\end{figure*}

{\bf Mixtures of threshold function families:} Threshold functions are commonly used to understand
the spread of new ideas/viral propogations in social networks, inspired by pioneering work in
sociology~\cite{Sc78,Gr78,KeKlTa03}. Think of two kinds of people (vertices) in a network. Some simply side with
the majority of their neighbors. Others are more resistant to change, and only take up a new belief
if all their neighbors believe it. 
We will first demonstrate our theorem on a synthetic distribution inspired by this application.
For simplicity of analysis (and to see all the math), set $K = 3$. The majority function is $\maj(y) = \sgn(\sum_i y_i)$
and the AND function $\al(y) = \sgn(\sum_i y_i + 2.5)$ (this is $-1$ iff all inputs are $-1$). Our distribution $\cF$
picks $\maj$ with probability $\beta$ and $\al$ with probability $1-\beta$. How much of the initial network
needs to have a new belief for it to propogate through the network? (And how is this sensitive to perturbation?)
Formally, what is the dynamics for initial distribution $\cD_\rho$? Think of a vertex state being $-1$ (TRUE) if that vertex currently
believes the new idea.
We start with the Fourier expansions of $\maj$ and $\al$. 
\begin{align*} & \maj(y) = \sum_i y_i/2 - y_1y_2y_3/2 \\
& \al(y) = 3/4 + \sum_i y_i/4 - \sum_{i \neq j} y_iy_j/4 + y_1y_2y_3/4
\end{align*}
We compute $\sigma_0(\cF) = 3(1-\beta)/4$, $\sigma_1(\cF) = 3\beta/2 + 3(1-\beta)/4 = 3(1+\beta)/4$,
$\sigma_2(\cF) = 3(\beta-1)/4$, and $\sigma_3(\cF) = -\beta/2 + (1-\beta)/4 = (1-3\beta)/4$.
From \Thm{rec},

\vspace{-5pt}

\begin{align*} \delta_{t+1} = (1-3\beta)\delta^3_{t}/4 & + 3(\beta-1)\delta^2_{t}/4 \\
& + 3(1+\beta)\delta_{t}/4 + 3(1-\beta)/4
\end{align*}
\vspace{-5pt}

Any fixed point is a root of the following polynomial $p(\delta)$ (which basically measures $\delta_{t+1} - \delta_t$).
Note that when $p(\delta_t) > 0$, then $\delta_{t+1} > \delta_t$ (and vice versa). 

\vspace{-5pt}
\begin{align*}
p(\delta) & = [(1-3\beta)\delta^3 + 3(\beta-1)\delta^2 + (3\beta-1)\delta + 3(1-\beta)]/4 \\
& = (\delta-1)(\delta+1)[(1-3\beta)\delta-3(1-\beta)]/4
\end{align*}
\vspace{-10pt}

This characterizes the limits of $\delta_t$ as $t \rightarrow \infty$ (assuming convergence).
The first two are trivial roots, since the all $-1$s and all $+1$s states are fixed points imbalances
for the Boolean network. The third root $3(1-\beta)/(1-3\beta)$ is a new valid
imbalance (in the range $(-1,1)$) only when $\beta > 2/3$. 

Now, we can explain the dynamics. (We ignore the trivial cases $\rho = -1,+1$.)
\begin{asparaitem}
	\item $\beta \leq 2/3$: The polynomial $p(z) > 0$ for any $z \in (-1,1)$. Hence, for any
	non-trivial starting distribution $\cD_\rho$, the Boolean network converges
	to the all $+1$s state. So the new belief will always die out. 
	\item $\beta > 2/3$: There exists a new unstable fixed point for the imbalance at $\delta^* = 3(1-\beta)/(1-3\beta)$.
	We have $p(z) > 0$ if $z > \delta^*$ and $p(z) < 0$ if $z < \delta^*$.  If $\rho > \delta^*$, the eventual
	state is all $+1$s. If $\rho < \delta^*$, the eventual state is all $-1$s. 
\end{asparaitem}

	To understand the sensitivity to bit flips, it is quite natural that for situations where
	$\delta_t$ converges to $-1$ or $+1$, the network is insensitive to perturbations.
	Calculations yield that $\Inf(\cF;-1)$ and $\Inf(\cF;+1)$ are $<1$.
	By \Thm{hamm}, the networks are quiescent.
	At $\rho = \delta^*$, $\I(\cF;\delta^*) = 3\beta(1-(\delta^*)^2)/2 + 3(1-\beta)(1-\delta^*)^2/4$.
	By some elementary algebra, $\I(\cF;\delta^*) > 1$ when $\beta > 2/3$. 
	Hence, for $\rho = \delta^*$, the dynamics are chaotic (again, this is expected).

We performed simulations on Boolean networks with $10^4$ nodes. For a given $\beta$, we vary the starting
distribution $\rho$ and measure the imbalances at $t = 20$. (This was averaged over 1000 runs.)
The results are in \Fig{maj-and}, where each colored line denotes a different choice of $\beta$. The 
predicted transition of $\delta^* = 3(1-\beta)/(1-3\beta)$ is denoted by the dashed line, coinciding
nicely with the experimental transition point. As expected
we see some fluctations (due to chaotic behavior at $\delta^*$) at the transition point.

{\bf Nested canalyzing functions:} For a real application, we consider the nested canalyzing functions of~\cite{Kauffman04}. 
(We provide a full description of this distribution in the supplement.) Previous work suggests that this distribution
is reflective of real biological networks and is quiescent. We can use our theorems
to validate the quiescence. Let us the consider the polynomial $\delta_{t+1}-\delta_t$.
For example at $K=5$, a technical calculation yields $p(\delta) = -0.001\delta^4 + 0.016\delta^3 - 0.11\delta^2 - 0.69\delta + 0.71$.
For $K=10$, $p(\delta) = -0.007\delta^4 + 0.012\delta^2 - 0.099\delta^2 - 0.7\delta + 0.71$. These polynomials
have a single stable root $\delta^* \approx 0.9$ in $[-1,+1]$.  
Even as $K$ varies, the root is quite stable, so that fixed point
imbalance is at least $0.9$ \emph{regardless} of the degree distribution. 

We perform experiments for varying degree distributions with $10^4$ nodes, and varying starting state distributions $\cD_\rho$.
(We show only the results for $K=5$ for space reasons.) In \Fig{nested-imb}, we plot the 
imbalance as a function of time for varying $\rho$. Observe that the imbalance \emph{always} converges to around $0.9$. \emph{This means that roughly 90\% of the nodes
converge to the $+1$ (FALSE) state.} The influence $\I(\cF;\delta^*)$ is roughly $0.3$, so the network is quiescent.
This is validated by the Derrida plot in \Fig{nested-hamm}, which plots average Hamming distance over time (for $\rho=0$).
We observe that the Hamming distance rapidly decays to $0$.

\vspace{-20pt}

\section{Acknowledgements}
\vspace{-10pt}
Sandia is a multiprogram laboratory operated by Sandia Corporation, a wholly owned subsidiary of Lockheed Martin Corporation, for the U.S. Department of Energy under contract DE-AC04-94AL85000.

\bibliographystyle{ieeetr}
\bibliography{bnfourier}

\newpage

\appendix

\centerline{\Large \bf Supplemental Material}
\bigskip

\section{Preliminaries and notation} \label{sec:prelims}

For convenience, we state and formalize many of the basic concepts already introduced in the main
body.

For $\rho \in [-1,1]$, we define
a biased distribution $\cD_\rho$ on $\cB$ as follows. The probability of $+1$
is $(1+\rho)/2$ and that of $-1$ is $(1-\rho)/2$. Note that expectation
is exactly $\rho$. We sometimes abuse notation and use $\cD_\rho$ to denote the product
distribution over $n$ bits. The uniform distribution is given by $\cD_0$.

We assume that there is a distribution $\cT$ on transfer functions.
Formally, this is a union of distributions $\cT_d$, where this family
only contains boolean functions that take $d$ inputs. For each vertex $v$
with indegree $d$, we first choose an independent function $f_v(y_1,y_2,\ldots,y_{d})$ 
from $\cT_{d}$. Randomly permute the in-neighbors of $i$ to get a list
$v_1, v_2, \ldots, v_{d}$. Assign the vertex $v_j$ to input $y_j$
of $\phi_i$. This gives us the transfer function for vertex $i$.

A convenient method for ignoring varying degrees is the following.
We assume that each vertex has an indegree of $\dm$, with neighbors
chosen randomly as before. Any function $\phi$ with less than $d < \dm$ inputs
can be extended to have $\dm$ inputs, where $\phi$ does not
depend on the new $\dm-d$ inputs. We now apply the same construction, where
there is a single distribution $\cT$ over input functions.

For a boolean network $\cN$, we use $f_t(x)$ to denote the total state after $t$ steps starting
with an initial state $x$. We use $f_{v,t}(x)$ to denote the (boolean) state at the vertex $v$.
Our aim is to understand $H_t = (1/n)\sum_{i = 1}^n \EX_{x \in \cD_\rho}[f_t(x) - f_t(\flip{x}{i})]$.
Meaning, we look at the expected Hamming distance over the starting state $x$ for a random
bit flip.
As proven in previous work, this is the same as $\frac{1}{n} \sum_{1 \leq u,v \leq n} \biasinf{\rho}{u}{f_{v,t}}$.
This is the average value
(over all vertices $v$) of $\sum_u \biasinf{\rho}{u}{f_{v,t}}$.
Since the construction of boolean networks is random where all vertices are symmetric,
in expectation, all these influence sums are the same. Hence, we will fix
a single vertex and focus on this sum.

\section{Fourier Analysis of Boolean Functions}

We will focus on functions of the form $f: \{-1,+1\}^K \longrightarrow \{-1,+1\}$.
We think of a function as a vector in $\RR^{2^K}$, which is just an
explicit representation of the truth table. 
The Fourier basis for Boolean functions (also called the Walsh-Hadamard basis) provides an alternate basis
to represent functions. 
\begin{definition} \label{def:four}
\begin{itemize}
	\item Let $S \subseteq [K]$. The \emph{parity on $S$} is the function 
	$\chi_S(y) = \prod_{i \in S} y_i$. Conventionally, the function $\chi_\emptyset$
	is a constant function that takes value $+1$.
	\item For $S \subseteq [K]$, define $\widehat{f}(S) := \EX_{x \in \cD_0}[f(y)\chi_S(y)] = 2^{-K} \sum_x f(y)\chi_S(y)$.
\end{itemize}

\end{definition}

The fundamental theorem is that the parities form an orthonormal basis for functions $f$
on the $\cB^d$. This gives the \emph{Fourier expansion of $f$}.

\begin{theorem} \label{thm:four} Every function $f: \{-1,+1\}^d \longrightarrow \RR$
is uniquely expressible as a linear combination of the parity functions.
Formally, $ f = \sum_{S \subseteq [K]} \widehat{f}(S) \chi_S$.
\end{theorem}

The influences are fundamentally connected to the Fourier expansion.

\begin{prop} \label{prop:inf} 
The value of $\biasinf{\rho}{i}{f}$ is equal to the following
three expressions.
\begin{itemize}
	\item $(1/4) \EX_{x \sim \cD_\rho}[(f(y) - f(y^{(i)}))^2]$
	\item $\EX_{x \sim \cD_\rho}[\big(\sum_{S \ni i} \wh{f}(S) \chi_{S \backslash i}(y)\big)^2]$
\end{itemize}
\end{prop}

\begin{proof} Since the probability distribution is always $\cD_\rho$, we drop the subscript $x \sim \cD_\rho$.
We have $\biasinf{\rho}{i}{f} = \Pr[f(y) \neq f(y^{(i)})]$. 
Observe that
$(f(y) - f(y^{(i)}))^2 = 4$ if $f(y) \neq f(y^{(i)})$ and zero otherwise. 
Hence, $4\cdot\biasinf{\rho}{i}{f} = \EX[(f(y) - f(y^{(i)}))^2]$.
We expand this expression.
\begin{eqnarray*} & & 4\cdot\biasinf{\rho}{i}{f} = \EX[(f(y) - f(y^{(i)}))^2] \\
& = & \EX[\Big(\sum_S \wh{f}(S) (\prod_{j \in S}y_j - \prod_{j \in S}y^{(i)}_j) \Big)^2] \\
& = & \EX[\Big(\sum_{S \ni i} \wh{f}(S) (y_i - y^{(i)}_i)\prod_{j \in S\backslash i}y_j \Big)^2] \ \ \ \ \textrm{(since for $j \neq i$, $y_j = y^{(i)}_j$)}\\
& = & 4\EX[\Big(\sum_{S \ni i} \wh{f}(S) \prod_{j \in S\backslash i}y_j\Big)^2] \ \ \ \ \textrm{(since $|y_i - y^{(i)}_i| = 2$)}
\end{eqnarray*}
\qed\\
\end{proof}

\ignore{

\begin{corollary} \label{cor:inf} The total unbiased influence of a function $\sum_i \biasinf{0}{i}{f}$
is $\sum_{S} |S| \wh{f}(S)^2$.
\end{corollary}

\begin{proof} By the last item in \Prop{inf}, $\biasinf{\rho}{i}{f} = \sum_{S \ni i, T \ni i} \wh{f}(S) \wh{f}(T) \rho^{|S \Delta T|}$.
When $\rho = 0$, any term corresponding to $|S \Delta T| \neq 0$ vanishes. Hence, what remains
is $\biasinf{0}{i}{f} = \sum_{S \ni i} \wh{f}(S)^2$. We sum this over all $i$ 
to get that the total influence is $\sum_i \sum_{S \ni i} \wh{f}^2(S)$. Observe that $\wh{f}(S)^2$
appears exactly $|S|$ times in the double summation, so this is $\sum_{S} |S| \wh{f}(S)^2$.
\qed\\
\end{proof}
}

\section{Deriving the recurrences} \label{sec:tree} 

Fix a vertex $v$. Let us consider the function $f_{v,t}$ for small $t \ll \log n$. 
Previous work tells us that we can assume (asymptotically) this is a rooted tree~\cite{Seshadhri11}.
We use $\nbdt{t}{v}$
to denote the $t$-step in-neighborhood of vertex $i$.

\begin{claim} \label{clm:g-tree} Fix a vertex $v$ and let
$t \leq (\log n)/(4\dm)$. The probability that 
the subgraph induced by $\nbdt{t}{v}$ is a directed tree
is at least $1 - 1/\sqrt{n}$.
\end{claim}

\textbf{The distribution $\cB_t$:} We define a distribution on Boolean networks that runs for $t$
steps on rooted trees with height $t$. This captures the $t$-neighborhood of $v$
based on \Clm{g-tree}. 
We take a $\dm$-ary directed tree rooted at $v$ of depth $t$ , with edges pointing towards the root $v$. 
For every internal node $u$, we choose a transfer function $\phi_u$ distributed according to $\cF$.
The leaves of the tree are the input nodes, collectively denoted as $x$. We will set the state
at leaf nodes from the distribution $\cD_{\rho}$. So $\delta_0 = \rho$ is the initial imbalance.

The Boolean network runs for $t$ steps
to yield the state at the root. Observe that for a vertex $u$ at height $h$, only the function $f_{u,h}$
is defined. 

We will use $v_1, v_2, \ldots$ to denote the children of $v$.
The Fourier expansion yields
the following claim. This innocuous statement is the heart of the analysis.

\begin{claim} \label{clm:fv} $f_{v,t} = \sum_{A \subseteq [\dm]} \wh{\phi_v}(A) \prod_{i \in A} f_{v_i,t-1}$
\end{claim}

\begin{proof} Suppose the state at $v_i$ is $y_i$. The state at $v$ is determined by
applying the transfer function $\phi_v$ on the states $(y_1, y_2, \ldots, y_K)$.
Using the Fourier expansion of $\phi_v$, we get the state at $v$
is $\sum_{A \subseteq [K]} \wh{\phi_v}(A) \prod_{i \in A} y_i$. The state $y_i$ is given
by the function $f_{v_i,t-1}$, and the state at $v$ is $f_{v,t}$.
\qed\\
\end{proof}

\subsection{The imbalance recurrence}

We derive a polynomial recurrence for $\bal_t$, the expected imbalance at a vertex after $t$ steps.
We have $\bal_t = \EX_{\cB_t}[\EX_{x \in \cD_\rho}[f_{v,t}(x)]]$.
For any $r$, remember that $\sigma_r = \EX_{\phi \sim \cF}[\sum_{C: |C| = r} \widehat{\phi}(C)]$.

\begin{theorem} \label{thm:bal} Let $\bal_t$ be the expected imbalance at time $t$.
For $t \geq 1$, $\bal_t$ evolves according to the following iterated polynomial map.
\begin{equation}
	\bal_t = \sum_{r \geq 0} \sigma_r \bal^r_{t-1} \label{eqn:bal}
\end{equation}
\end{theorem}

\begin{proof} 
We take expectations of the formula in \Clm{fv}. (Verbal explanation follows.)
\begin{eqnarray*}
\bal_t = \EX_{x,\cB_t}[f_{v,t}(y)] & = & \EX_{x,\cB_t}[\sum_{A \subseteq [\dm]} \wh{\phi_v}(A) \prod_{i \in A} f_{v_i,t-1}(y)] \\
& = & \sum_{A \subseteq [\dm]} \EX_{x,\cB_t}\Big[\wh{\phi_v}(A) \prod_{i \in A} f_{v_i,t-1}(y)\Big] \\
& = & \sum_{A \subseteq [\dm]} \EX_{\cF} [\wh{\phi}(A)] \prod_{i \in A} \EX_{x,\cB_{t-1}} [f_{v_i,t-1}(y)]
\end{eqnarray*}
The second line is just linearity of expectation. The final line is obtained through independence. Note that $\phi_v$
is independent of the Boolean networks rooted at the $v_i$s. These Boolean networks are also independent of each other.
Hence, the expectation of the product is the product of expectations. The function $\phi_v$ is a random
function $\phi$ chosen from $\cF$. Because of the recursive
construction, the distribution of $\cB_t$ rooted at $v$ induces the distribution of $\cB_{t-1}$ rooted at the $v_i$s.
Now, observe that $\EX_{x,\cB_{t-1}} [f_{v_i,t-1}(y)] = \bal_{t-1}$. 

Plugging this in and collecting all terms corresponding to sets of the same size,
\begin{eqnarray*} 
\bal_t & = & \sum_{A \subseteq [\dm]} \bal^{|A|}_{t-1} \EX_{\cF} [\wh{\phi}(A)] \\
& = & \sum_{r \geq 0} \bal^{r}_{t-1} \sum_{A: |A| = r} \EX_\cF[\wh{\phi}(A)] 
= \sum_{r \geq 0} \sigma_r \bal^r_{t-1}
\end{eqnarray*}
\qed\\
\end{proof}

\subsection{The spreading of perturbations}

We focus on $\Iex_t(\rho_0)$, the expected average (over all nodes) influence of a node
at $t$-steps, when the initial distribution is $\cD_{\rho_0}$.
By the tree approximation, it suffices to focus on the node $v$ and consider the
distribution $\cB_t$.
We can express $\Iex_t(\rho_0)$ as follows.

By the tree approximation, 
$H_t = \EX_{\cB_t}[\sum_\ell \Inf_\ell(f_{v,t}); \rho]$ (where $\ell$ is over all leaves). In words, we look at the $\rho$-biased influence summed over all leaves.
For convenience, we will drop the time/height subscript and simply write $f_u$ instead of $f_{u,h}$.

\begin{theorem} $H_t = \prod_{0 \leq h < t} {\biastot{{\bal_h}}{\cF}}$
\end{theorem}

\begin{proof} Partition the leaves into subsets $S_1, S_2, \ldots, S_{\dm}$, where
$S_i$ contains all leaves that are descendants of $v_i$. 
Focus on a leaf $\ell \in S_1$.
By \Prop{inf} and \Clm{fv},
\begin{eqnarray*} & & \EX_{\cB_t} [\Inf_{\ell}(f_{v}; \rho)] \\
& = & (1/4)\EX_{\cB_t,x \sim \cD_{\rho}}[(f_{v}(x) - f_{v}(x^{(\ell)}))^2] \\
& = & (1/4)\EX_{\cB_t,x \sim \cD_{\rho}}[\Big\{\sum_A \wh{\phi}_v(A) (\prod_{i \in A} f_{v_i}(x) - \prod_{i \in A} f_{v_i}(x^{(\ell)}))\Big\}^2] 
\end{eqnarray*}
Observe that for $i \neq 1$, $f_{v_i}(x) = f_{v_i}(x^{(\ell)})$. (This is because $\ell$
is not in the subtree of $v_i$.)
In the summation
above, only the terms corresponding to $A \ni 1$ are non-zero. Expanding further,
\begin{eqnarray*}
& & \big\{\sum_{A \ni 1} \wh{\phi}_v(A) (\prod_{i \in A} f_{v_i}(x) - \prod_{i \in A} f_{v_i}(x^{(\ell)}))\big\}^2 \\
& = & \Big\{\sum_{A \ni 1} \wh{\phi}_v(A) \big(\prod_{\substack{i \in A\\i \neq 1}} f_{v_i}(x)\big)(f_{v_{1}}(x) - f_{v_{1}}(x^{(\ell)}))\Big\}^2 \\
& = & \big(f_{v_{1}}(x) - f_{v_{1}}(x^{(\ell)})\big)^2
\Big\{\sum_{A \ni 1} \wh{\phi}_v(A) \prod_{\substack{i \in A\\i \neq 1}} f_{v_i}(x)\Big\}^2
\end{eqnarray*}
Each $f_{v_i}$ is defined over disjoint parts of the underlying tree with disjoint inputs.
Hence, when we take the expectation $\EX_{\cB_t,x}$ over the product, we get the product
of expectations. Moreover, $(1/4) \EX_{\cB_t,x \sim \cD_{\rho}}[\big(f_{v_{1}}(x) - f_{v_{1}}(x^{(\ell)})\big)^2]$
is exactly $\EX_{\cB_{t-1}} [\biasinf{\rho}{\ell}{f_{v_{1}}}]$. 

The random variable $f_{v_i}(x)$ is in $\{-1,+1\}$ and $\EX_{\cB_{t-1},x \sim \cD_{\rho}}[f_{v_i}(x)] = \bal_{t-1}$.
Hence, it is distributed as $\cD_{\bal_{t-1}}$. Taking expectations over $\cB_{t-1},x$, setting $y_i = f_{v_i}(x)$
and \Prop{inf},
\begin{eqnarray*} 
& & \EX_{\cB_{t},x \sim \cD_{\rho}}[\big\{\sum_{A \ni 1} \wh{\phi}_v(A) \prod_{\substack{i \in A\\i \neq 1}} f_{v_i}(x)\big\}^2] \\
& = & \EX_{\phi \sim \cF, y \sim \cD_{\bal_{t-1}}} [\big\{\sum_{A \ni 1} \wh{\phi}_v(A) \prod_{\substack{i \in A\setminus 1}} y_i\big\}^2] \\
& = & \EX_{\phi \sim \cF}[\biasinf{\bal_{t-1}}{1}{\phi}] 
\end{eqnarray*}
In general, for $\ell \in S_i$, we get 
$\EX_{\cB_t} [\Inf_{\ell}(f_{v}; \rho)] = \EX_\cF[\biasinf{\bal_{t-1}}{i}{\phi}] \EX_{\cB_{t-1}} [\biasinf{\rho}{\ell}{f_{v_{i}}}]$.
We combine all our observations. 
\begin{eqnarray*}
H_t & = & \sum_\ell \EX_{\cB_t}[\Inf_{\ell}(f_{v,t}; \rho)] \\
& = & \sum_{i=1}^\dm \sum_{\ell \in S_i} \EX_{\cB_t}[\Inf_{\ell}(f_{v,t})] \\
& = & \sum_{i=1}^{\dm} \EX_{\cF} [\biasinf{\bal_{t-1}}{i}{\phi}] \sum_{\ell \in S_i} \EX_{\cB_{t-1}} [\biasinf{\rho}{\ell}{f_{v_{i}}}]  \\
& = & \sum_{i=1}^{\dm} \EX_{\cF} [\biasinf{\bal_{t-1}}{i}{\phi}] \EX_{\cB_{t-1}}[ \sum_{\ell \in S_i} \biasinf{\rho}{\ell}{f_{v_{i}}}] \\
& = & H_{t-1} \sum_{i=1}^{\dm} \EX_{\cF} [\biasinf{\bal_{t-1}}{i}{\phi}] \\
& = & H_{t-1} \cdot {\biastot{\bal_{t-1}}{\cF}} 
\end{eqnarray*}
Uncoiling the recurrence yields the theorem.
\qed\\
\end{proof}

\newpage
\section{Nested canalyzing functions}

For completeness, we describe this distribution.
Fix positive integer
$\alpha$ and a series of canalyzing input values $c_1, c_2, \ldots, c_K$ and $d_1, d_2, \ldots, d_K, d_{def}$ (where
each of these is in $\{-1,+1\}$). The function is defined as follows:
$$ f(x) = 
\begin{cases}
	d_1 & \text{if } y_1 = d_1 \\
	d_2 & \text{if } y_1 \neq d_1 \text{and } y_2 = d_2 \\
	\vdots & \\
	d_K & \text{if } y_1 \neq d_1, y_2 \neq d_2, \ldots, y_{K-1} \neq d_{K-1} \text{and } y_K = d_k \\
	d_{def} & \text{otherwise}
\end{cases}
$$
For any parameter $\alpha > 0$, the distribution is given by $\Pr[c_i = -1] = \Pr[d_i = -1] = \exp(-\alpha/2^i)/(1+\exp(-\alpha/2^i))$.
Kauffman et al suggest that $\alpha = 7$ is reflective of real biological networks, and corresponding boolean networks
are quiescent.

\end{document}